\documentclass[letterpaper, 10 pt, conference]{ieeeconf}

\IEEEoverridecommandlockouts

\usepackage{cite}

\usepackage{amsmath,amssymb,amsfonts, mathabx}
\usepackage{mathtools}
\usepackage{authblk}
\usepackage{booktabs} 
\usepackage{indentfirst}
\usepackage{times}
\usepackage{booktabs}   
\usepackage{amsfonts}      
\usepackage{nicefrac}     
\usepackage{microtype}

\usepackage{amsthm}
\usepackage{hyperref}
\usepackage[utf8]{inputenc} 
\usepackage[T1]{fontenc}  
\usepackage[ruled]{algorithm2e}
\usepackage{graphicx}
\usepackage{xspace}
\usepackage{cleveref}
\usepackage{textcomp}
\usepackage{xcolor}
\def\BibTeX{{\rm B\kern-.05em{\sc i\kern-.025em b}\kern-.08em
    T\kern-.1667em\lower.7ex\hbox{E}\kern-.125emX}}

\newtheorem{theorem}{Theorem}
\numberwithin{theorem}{section}
\newtheorem{lemma}{Lemma}
\newtheorem{assumption}{Assumption}

\newtheorem{corollary}{\textit{Corollary}}

\numberwithin{corollary}{section}

\newtheorem{definition}{Definition}

\newcommand{\IMPC}{{\textsc{InstructMPC}}\xspace}
\newcommand{\removelatexerror}{\let\@latex@error\@gobble}

\begin{document}

\title{\LARGE \bf \IMPC: A Human-LLM-in-the-Loop Framework for Context-Aware Control}

\author{Ruixiang Wu$^{1}$, \ Jiahao Ai$^{2}$, \  Tongxin Li$^{1,*}$%
\thanks{This work was supported by the Guangdong Basic and Applied Basic Research Foundation (No. 2025A1515011311); the National Natural Science Foundation of China (NSFC) under Grant No. 72301234; the `1+1+1' Joint Funding Program (Key Scientific Research Projects); the Guangdong Provincial Key Laboratory of Mathematical Foundations for Artificial Intelligence (No. 2023B1212010001), and the PengCheng Peacock Supporting Scientific Research Fund.}
\thanks{
  Ruixiang  Wu and Tongxin Li are with School of Data Science, The Chinese University of Hong Kong, Shenzhen, Guangdong, 518172, China~(E-mails: {\tt ruixiangwu@link.cuhk.edu.cn, litongxin@cuhk.edu.cn})
}
\thanks{Jiahao Ai is with School of Mathematical Sciences, Peking University, Beijing, 100871, China~(E-mail: {\tt 2100010637@stu.pku.edu.cn})} 
}

\maketitle
\thispagestyle{empty}
\pagestyle{empty}

\begin{abstract}
Model Predictive Control~(MPC) is a powerful control strategy widely utilized in domains like energy management, building control, and autonomous systems. However, its effectiveness in real-world settings is challenged by the need to incorporate context-specific predictions and expert instructions, which traditional MPC often neglects. We propose \IMPC, a novel framework that addresses this gap by integrating real-time human instructions through a Large Language Model~(LLM) to produce context-aware predictions for MPC. Our method employs a Language-to-Distribution~(L2D) module to translate contextual information into predictive disturbance trajectories, which are then incorporated into the MPC optimization. Unlike existing context-aware and language-based MPC models, \IMPC enables dynamic human-LLM interaction and fine-tunes the L2D module in a closed loop with theoretical performance guarantees, achieving a regret bound of $O(\sqrt{T\log T})$ for linear dynamics when optimized via advanced fine-tuning methods such as Direct Preference Optimization~(DPO) using a tailored loss function.

\end{abstract}

\section{Introduction}

Model Predictive Control~(MPC) is a versatile control strategy widely applied in domains such as energy management~\cite{10886828, East_2018,li2021learning,li2025learning}, building control~\cite{drgovna2020all}, and autonomous driving~\cite{batkovic2023experimental, 10886377}. These applications leverage MPC’s ability to optimize system's behavior over a predictive horizon, making it a cornerstone of modern control engineering.

In real-world scenarios, however, the predictions used in MPC are not solely derived from static models. Instead, they often depend on concrete, context-specific scenarios or are determined by experts to be near-optimal. For example, in building energy management of a university campus, MPC must anticipate and adapt to scenarios such as impending paper submission deadlines causing computational load surges, scheduled public events increasing demand, and forecasted weather changes affecting PV generation. This context dependency underscores a critical question: 
“\textit{How can we incorporate in-context instructions into MPC?}” 
Key challenges in addressing this issue revolve around adapting to evolving human instructions and ensuring generalizability in changing environments. The MPC system must be capable of dynamically adjusting its behavior in response to shifting human inputs and contexts, which demands a high level of flexibility and responsiveness to maintain effectiveness. Furthermore, the learning-based integration of these instructions must be designed to adapt to time-varying conditions, allowing the system to not only respond to immediate changes but also to generalize its learning effectively, ensuring consistent performance over time. This combination of adaptability and robust generalization is essential for overcoming the complexities presented by dynamic human interactions and fluctuating environmental factors.

To tackle these challenges, we introduce \IMPC, a novel framework that employs a human-LLM-in-the-loop approach to form a closed-loop system. \IMPC takes in-context instructions from humans, processes them through a Large Language Model (LLM) to generate context-aware predictions, and feeds these into the downstream MPC module. During online implementation, the model loss from MPC is propagated back to the LLM, which is fine-tuned using techniques such as Proximal Policy Optimization (PPO)~\cite{schulman2017proximalpolicyoptimizationalgorithms} and Direct Preference Optimization (DPO)~\cite{rafailov2023direct} to enhance adaptability and generalizability across scenarios.

Recent advancements in MPC have explored contextual and language-based enhancements to improve performance in autonomous systems and personalized applications.
\paragraph{Context-aware MPC for autonomous systems} In autonomous systems, context-aware MPC leverages environmental and semantic information to enhance control performance. For instance, semantically informed MPC~\cite{goel2023semantically} employs deep neural networks to encode semantic labels (e.g., ``find a couch'') into cost maps, guiding navigation in unknown environments with continuous control, achieving an improvement in success rates through mid-level visual cues. However, it focuses on discrete, predefined contexts without real-time human input. Similarly,~\cite{frohlich2022contextual} adapts MPC parameters to environmental contexts (e.g., rain) using contextual Bayesian optimization. While effective for model and objective tuning, it lacks mechanisms for incorporating real-time instructions from humans. Additionally, designed for crowded environments, the framework in~\cite{stefanini2024efficient} integrates contextual cues like human poses and activities into MPC. It prioritizes computational efficiency but does not generalize to in-context instructions beyond navigation tasks.

\paragraph{Language-Based MPC Personalization} Language-based approaches harness natural language to personalize MPC systems for improved usability and flexibility. For instance, the LanguageMPC in~\cite{sha2023languagempc} utilizes LLMs as decision-makers in autonomous driving, combining LLM reasoning with low-level controllers to enhance safety, efficiency, and generalizability in complex scenarios, though it does not support real-time human interaction. Meanwhile, ChatMPC~\cite{miyaoka2024chatmpc} uses BERT to classify natural language into a set of predefined circumstances, each linked to a predefined modification of MPC parameters; however, it is constrained by its predefined updating rule, and also lacks state-of-the-art human-LLM interaction and theoretical performance guarantees.

% Meanwhile, the ChatMPC~\cite{miyaoka2024chatmpc} processes natural language via BERT to update MPC specifications dynamically, reducing the data burden for personalization; however, it lacks state-of-the-art human-LLM interaction and theoretical performance guarantees. 

The \IMPC framework integrates real-time human instructions through a closed-loop human-LLM system for dynamic adaptation and provides a theoretical performance guarantee for linear dynamics using DPO, ensuring features absent in existing context-aware and language-based approaches. Our contributions are two-fold:
\begin{enumerate}
    \item \textbf{\IMPC Framework}: \IMPC dynamically tunes MPC by integrating real-time, human-provided in-context instructions through a Language-to-Distribution (L2D) module (e.g., LLMs) and fine-tuning the L2D module in a closed loop. The novelty lies in enabling the MPC controller to leverage contextual information often ignored in traditional control models. By transforming high-level instructive information into predicted  disturbance trajectories, \IMPC enhances the adaptability of classic MPC.
    \item \textbf{Theoretical Guarantees}: For linear dynamics, we provide theoretical performance (Theorem~\ref{thm:main}) guarantees for closed-loop LLM fine-tuning using Direct Preference Optimization (DPO). Unlike traditional online optimization methods that directly minimize regret, our approach decouples the DPO loss from the regret objective, which relies on MPC cost information and is challenging to compute online due to unknown future disturbances. 
    In particular, Theorem~\ref{thm:main} suggests a tailored loss function for LLM fine-tuning, achieving a regret bound of $O(\sqrt{T \log T})$ as derived in Corollary~\ref{cor}.
\end{enumerate}

\section{Preliminaries and Problem Formulation} \label{sec:problem}

Denote $[T]\coloneqq \{0,1,\ldots,T-1\}$. We consider a finite horizon linear dynamical system
\begin{equation}\label{eq:linear sys}
    x_{t + 1} = A x_t + B u_t + w_t, \quad t\in [T]\
\end{equation}
where $x_t \in \mathbb{R}^n, u_t \in \mathbb{R}^m$ denote the system state and control input at each time $t\in [T]$; $A \in \mathbb{R}^{n \times n}$ and $B \in \mathbb{R}^{n \times m}$ are system matrices. The disturbance $w_t \in \mathbb{R}^n$ is unknown to the controller at time $t\in [T]$.  The controller's goal is to minimize the following quadratic costs:
\begin{subequations}
\begin{align}
\label{eq:quadratic_costs}
J^{\star}\coloneqq\min_{(u_t:t\in [T])}&\sum_{t=0}^{T-1}(x_t^{\top}Qx_t+u_t^{\top}Ru_t)+x_T^{\top}Px_T,\\
\label{eq:sys_constraints}
&\text{subject to }\eqref{eq:linear sys},
\end{align}
\end{subequations}
where $Q, R \succ 0$ are positive definite matrices, and $P$ is the solution to the following discrete algebraic Riccati equation (DARE),\footnote{\scriptsize To simplify the presentation of our main results, we fix the terminal cost in~\eqref{eq:quadratic_costs} to be $P$. The arguments extend to more general terminal costs as well, since the overall cost only differs by an $O(1)$ term.}
\begin{equation} \label{eq:dare}
    P = Q+ A^{\top}PA - A^{\top}PB(R+B^{\top}PB)^{-1}B^{\top}PA.
\end{equation}
The disturbance $w_t$ is bounded by a constant $W>0$~(i.e. $w_t\in\mathcal{W}\coloneqq \{w:\|w\|\leq W\}$)  for all $t\in [T]$.
Futhermore, the system $(A,B)$ is stabilizable~\cite{dullerud2013course}.
\begin{assumption}\label{asp:mpc}
There exists $K\in\mathbb{R}^{m\times n}$ such that $A-BK$ has a spectral radius $\rho$ less than $1$.\footnote{\scriptsize Thus, the Gelfand’s formula implies that there exist $C>0$, $\rho\in (0,1)$ such that $\|A-BK\|^{t}\leq C\rho^t$ for all $t\geq 0$. }
\end{assumption}

\subsection{Model Predictive Control} 

Fix a prediction horizon to be an integer $k$ and define  $\mathcal{T} \coloneqq \min\{t + k - 1, T - 1\}$.
We consider the setting in which, at each time $t\in [T]$, $k$-step predictions 
$
\hat{w}_{t:\mathcal{T}|t} \coloneqq (\hat{w}_{t|t}, \ldots, \hat{w}_{\mathcal{T}|t})\in\mathcal{W}^k
$ corresponding to future disturbances $(w_{t},\ldots,w_{\mathcal{T}})$
are available. 

In this context, model predictive control naturally incorporates these predictions into its optimization framework. At each time $t\in [T]$, the controller solves the following optimization problem given the current state $x_t$ and disturbance predictions in $\hat{w}_{t:\mathcal{T}|t}$:
\begin{align} \label{eq:mpc_formulation}
    u_{t:\mathcal{T}}^{\textsc{MPC}} &\coloneqq  \arg\min_{u} \Big (\sum_{\tau = t}^{\mathcal{T}} (x_{\tau}^{\top} Q x_{\tau} + u_{\tau}^{\top} R u_{\tau}) + x^{\top}_{\mathcal{T} + 1} P x_{\mathcal{T} + 1}\Big) \notag\\
    & \text{s.t. } x_{\tau + 1} = A x_{\tau} + B u_{\tau} + \hat{w}_{\tau|t}, \forall t \leq \tau < \mathcal{T}.
\end{align}
Only the first element of $u_{t:\mathcal{T}}^{\textsc{MPC}}$ is applied as a control input to the system, and the rests are discarded. 
However, in many real-world scenarios, system disturbances are hard to predict~\cite{binder2019improved, castillo2020predictingfuturestatedisturbed}. In this paper, we propose a novel framework that leverages contextual information to predict future disturbances and incorporates these predictions into the decision-making process. 

\subsection{Contextual MPC with Model Fine-Tuning}
Our goal is to improve the performance of the classic MPC defined in~\eqref{eq:mpc_formulation} by incorporating external contextual information at time $t\in [T]$, denoted by $c_t\in\mathcal{C}$ to predict future disturbances $\hat{w}_{t:\mathcal{T}|t}\in\mathcal{W}^k$.\footnote{\scriptsize The predicted disturbance trajectory $\hat{w}_{t:\mathcal{T}|t}$ is generated in real-time, operating on the same time scale as the MPC decision-making process.} Hence, the main challenge is to integrate a Neural Network (NN) $g_{\theta}: \mathcal{C} \to \mathcal{W}^k$ to detokenize the contextual information into predictions for the downstream MPC, and refine these predictions by fine-tuning the model parameter $\theta\in\Theta$ after the real disturbances are revealed.  Let $\|\cdot\|$ denote the $\ell_2$-norm of a vector. The hypothesis set $\Theta$ has bounded diameter so that $\|\theta-\theta'\|\leq D$ for all $\theta,\theta'\in \Theta $ and some constant $D>0$. 
We impose the following assumptions on the NN model $g_{\theta}$, which serve as approximations for widely used architectures, including transformers~\cite{vaswani2017attention} and input convex neural networks~\cite{amos2017input}.
% such as input convex neural networks~\cite{amos2017input}
% Formally, the regret is defined as:
%     \begin{equation*}
%         \textsf{Regret} \coloneqq J(\theta_1,...,\theta_T)-J(\theta^{\star})
%     \end{equation*}
\begin{assumption}\label{asp:convex}
The NN model $g_\theta$ is differentiable and affine in $\theta$ and its gradient satisfies $\|\nabla_\theta g_\theta\|\leq L$ for some constant $L>0$.
    % $g_\theta(c)$ is convex in $\theta$;    $\nabla_\theta g_\theta(c)$ is bounded by $L$
\end{assumption}
A hyper-parameter $\theta_t$ is updated at  $t \in [T]$. 
We evaluate performance via the following \emph{regret}, casting the problem as an online optimization over $\theta_{1:T}\coloneqq (\theta_t : t \in [T])$:
\begin{align}
    \label{eq:regret}
    J(\theta_{1:T}) - \min_{\theta \in \Theta} J(\theta),
\end{align}
where $J(\theta_1,\ldots,\theta_{T})$ and $J(\theta^{\star})$ are the quadratic costs defined in~\eqref{eq:quadratic_costs} induced by the MPC actions in~\eqref{eq:mpc_formulation} corresponding to the hyper-parameters $\theta_1,\ldots,\theta_{T}$ and the optimal hyper-parameter $\theta^{\star}$, respectively. In other words, the regret in~\Cref{eq:regret} measures the difference between the cumulative loss incurred by $(\theta_t : t \in [T])$, and the cumulative loss incurred by the optimal $\theta^{\star}$ in hindsight. Next, we introduce the \IMPC framework for contextual MPC.

\section{InstructMPC} 
\label{sec:system}

In this section, we introduce \IMPC, a framework that enables MPC controller to leverage external contextual information to effectively manage environmental disturbances. Fig.~\ref{fig:system_framework} presents the overall architecture of \IMPC. In the following subsections, we explain the core component of \IMPC, Language to Distribution~(L2D) module in details.

\begin{figure}[h]
    \centering
    \includegraphics[width = 0.9\linewidth]{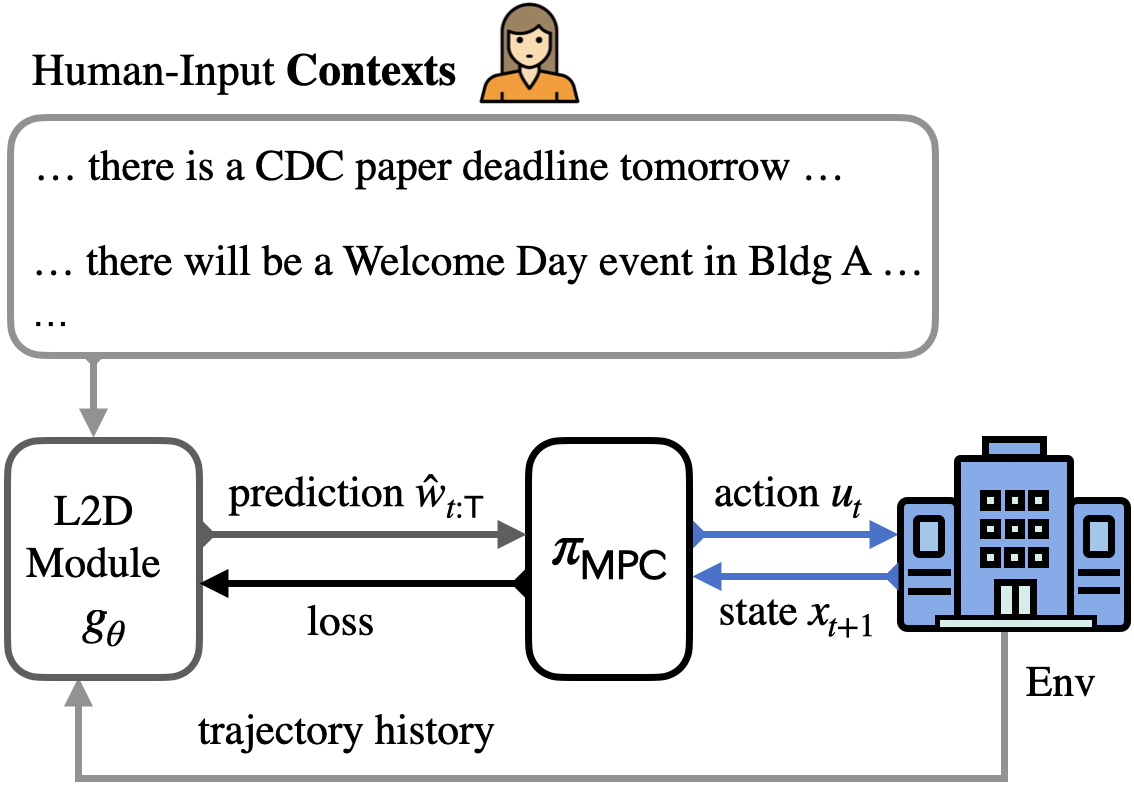}
    \caption{System framework of \IMPC. The blue lines represent interactions between \IMPC and the environment where \IMPC receives the state $x_t$, and outputs the control input $u_t$. The black lines represents the \textit{information loop}, within which external contextual information is passed to the L2D module $g_{\theta}$ to produce predicted disturbances $\hat{w}_{t:\mathcal{T}|t}$. Then, the MPC controller $\pi_{\mathrm{MPC}}$ utilizes $\hat{w}_{t:\mathcal{T}|t}$ and the current state $x_t$ to determine a control input $u_{t}^{\mathrm{MPC}}$ via~\eqref{eq:mpc_formulation}. After executing the MPC control input $u_{t}^{\mathrm{MPC}}$, the environment reveals the true disturbance $w_t$. The discrepancy between $w_t$ and $\hat{w}_t$ is then sent back to the L2D module as a loss signal.}
    \label{fig:system_framework}
\end{figure}

\subsection{Language to Distribution~(L2D) Module} \label{sec:l2d}
During the execution of \IMPC, the L2D module $g_{\theta}$ predicts future environmental disturbances $\hat{w}_{t:\mathcal{T}|t} = (\hat{w}_{t|t},\ldots ,\hat{w}_{\mathcal{T}|t})$ based on contextual information. After each control step, the module refines its parameters $\theta$ by comparing these predictions with the actual revealed disturbances.

At each time $t$, given contextual information $c_t \in \mathcal{C}$, the module first assigns a probability $p(s | c_t)$ to each scenario $s \in \mathcal{S}$. For each scenario $s$, the module retrieves a disturbance trajectory $w^{s}_{t:\mathcal{T}} = (w^s_{t|t}, \ldots, w^s_{\mathcal{T}|t})$ based on historic data. The predicted disturbance trajectory is computed by taking a weighted combination of trajectories $(w^{s}_{t:\mathcal{T}} : s \in \mathcal{S})$:
\begin{equation*}
    \hat{w}_{t:\mathcal{T}|t} = \sum_{s \in \Omega} p(s | c_t)  w^{s}_{t:\mathcal{T}}.
\end{equation*}

\subsection{\IMPC Framework} 

\begin{algorithm}[h]
\label{alg:main}
\caption{\IMPC}
\label{alg:lqr}

\For{$t = 0, \dots, T-1$}{
$\mathcal{T} \leftarrow \min(t + k - 1, T - 1)$\\
\textbf{Get predictions} from $g_{\theta}$ with instruction $c_t$:
\vskip 0.2em
$(\hat{w}_{t|t}, \ldots, \hat{w}_{\mathcal{T}|t})\gets g_{\theta_t}(c_t)$
    \vskip 0.2em
    Generate an action $u_t$ by solving MPC in~\eqref{eq:mpc_formulation}
    \vskip 0.2em 
    \textbf{Update} 
    \begin{align*}
        x_{t+1} &= A x_t + B u_t + w_t\\
        \theta_{t+1} &= \theta_t-\eta_t\nabla_\theta L_{t-k+1}(\theta_{t-k+1})
    \end{align*}
}

\end{algorithm}

The L2D module is a neural network $g_{\theta}: \mathcal{C} \to \mathcal{W}^k$, which is updated at every time $t$ through a loss function $L_{t-k+1}:\Theta\times\mathcal{W}^{2k}\rightarrow\mathbb{R}$, which also depends on disturbances $w_{t:\mathcal{T}}$ and predictions $\hat{w}_{t:\mathcal{T}|t}$. An example of such a loss function can be found in Corollary~\ref{cor} in Section~\ref{sec:main}. At every time $t \in [T]$, the parameter $\theta$ of L2D module is updated according to:
\begin{equation} \label{eq:theta_update}
    \theta_{t+1}=\theta_t-\eta_t\nabla_\theta L_{t-k+1}(\theta_{t-k+1}),
\end{equation}
where $\eta_t\nabla_\theta L_{t-k+1}$ denotes the gradient of $L_{t-k+1}$ with respect to $\theta$ and $\eta_t$ is a time-varying learning rate.
The gradient descent rule above is widely used to fine-tune NN models, such as Direct Preference Optimization~(DPO)~\cite{rafailov2023direct}, Proximal Policy Optimization~(PPO)~\cite{schulman2017proximalpolicyoptimizationalgorithms}, and Supervised Fine-Tuning~(SFT)~\cite{ouyang2022traininglanguagemodelsfollow}.
The overall procedure of \IMPC is summarized in Algorithm~\ref{alg:main}. Since $\nabla_\theta L_t(\theta_{t})$ is available only after $(w_t,\ldots,w_{\mathcal{T}})$ are revealed at time $t+k$, a delayed gradient is used in the update rule~\eqref{eq:theta_update}. For convenience, we define $L_{t-k+1}=0$ for $t<k$. In practice, when $t<k$, the NN model does not update due to insufficient observed data.

% In the next subsection, we explain how the MPC controller utilize the predicted disturbances $\hat{w}_{t:\mathcal{T}|t}$ to make control decisions.

\section{Theoretical Guarantee}
\label{sec:main}

Our main result in this section explores the online optimization of $(\theta_t:t\in [T])$ in the online fine-tuning process. We make the following standard assumption on the loss function $L_t$.
% \begin{definition}[Regret] \label{def:regret}
% \end{definition}
% We aim to derive an upper bound for the regret $J(\theta_1,...,\theta_T)-J(\theta^{\star})$, where $J(\theta_1,...,\theta_T)$ denotes the cumulative MPC cost incurred by Algorithm~\ref{alg:main} when the parameter $\theta$ is sequentially updated to $\theta_t$ at each step $t$, and $J(\theta^{\star})$ represents minimum achievable MPC cost when $\theta$ is fixed to its optimal value $\theta^{\star}$ throughout Algorithm~\ref{alg:main}.
\begin{assumption}\label{asp:lg}
 The gradient $\nabla_\theta L_t(\theta)$ is uniformly bounded, i.e., there exists $G>0$ such that by $\|\nabla_\theta L_t(\theta)\|\leq G$ for all $\theta\in\Theta$.
\end{assumption}

% Note that the loss function used for updating $g_\theta$ may not be aligned with regret-optimal objective that requires the knowledge of future ground-truth predictions $(w_t:t\in [T])$ in~\eqref{eq:mpc_formulation}. To quantify this mismatch, we define the \textit{loss discrepancy} below.

% Note that the loss function used for updating $g_\theta$ may not be aligned with the regret-optimal objective that requires the knowledge of future ground-truth predictions $(w_t:t\in [T])$ in~\eqref{eq:mpc_formulation}. Therefore, the main challenge in \IMPC framework is to ensure that adjusting the parameters in the upstream L2D module $g_\theta$ leads to provably improved performance for downstream MPC controller. To address the challenge, we introduce a novel \textit{loss discrepancy} below to capture the divergence between the update direction~(i.e., the gradient) suggested by the L2D module prediction loss $L$, and the update direction suggested by the true decision loss (see~\Cref{eq:mpc_formulation}).

A fundamental challenge is that the surrogate loss function used to train $g_\theta$ does not match the true, regret-optimal objective~\eqref{eq:mpc_formulation}, which is unknown at training time. This misalignment creates the primary difficulty, i.e., ensuring that training the upstream L2D module $g_\theta$ actually improves the downstream MPC controller's performance. Our solution is to define and analyze the \textit{loss discrepancy}, a term that quantifies the divergence between the gradient of the L2D module surrogate loss $L$ and that of the true decision loss in~\Cref{eq:mpc_formulation}.

% In line with the principle of decision focused learning~\cite{Mandi_2024}, the core challenge is not simply minimize L2D module's prediction error, but to produce predictions that are most useful for the downstream MPC controller. 

\begin{definition}
\label{def:df}
The loss discrepancy between two  loss functions $L_1(\theta)$ and $L_2(\theta)$, $\emph{\texttt{LD}}(L_1,L_2)$ is defined as 
\begin{equation*}
    \emph{\texttt{LD}}(L_1,L_2) \coloneqq \sup_{\theta \in \Theta}\left\|\nabla_\theta L_1(\theta)-\nabla_\theta L_2(\theta)\right\|.
\end{equation*}
\end{definition}

The theorem below provides a regret bound~(see~\Cref{eq:regret}) for the proposed \IMPC, which interacts with the dynamic system described via~\Cref{eq:linear sys}.
\begin{theorem}\label{thm:main}
    Under Assumption \ref{asp:convex},\ref{asp:lg}, if the learning rate $\eta_t$ is non-increasing, then
\begin{align}
\nonumber
    J(\theta_{1:T}) -& J(\theta^{\star})
    \leq \frac{D^2}{2\eta_{T-1}}+\left(k-\frac{1}{2}\right)G^2\sum_{t=0}^{T-1}\eta_t\\
    \label{eq:loss_discrepancy_bound}
    &+D\sum_{t=0}^{T-1}\emph{\texttt{LD}}(L_t,\psi_t^\top H\psi_t)+(k-1)GD.
\end{align}
Furthermore, if we choose $\eta_t=\frac{D}{G\sqrt{2(2k-1)(t+1)}}$,
\begin{align*}
    J(\theta_{1:T})-&J(\theta^{\star})\leq 2GD\sqrt{(k-\frac{1}{2})T}\\
    &+D\sum_{t=0}^{T-1}\emph{\texttt{LD}}\left(L_t,\psi_t^\top H\psi_t\right)+(k-1)GD,
\end{align*}
where $H$ is defined in Lemma~\ref{lem:regret}, and we define
\begin{equation*}
    \psi_t(\theta) \coloneqq \sum_{\tau=t}^{T-1}\left(F^\top\right)^{\tau-t}Pw_\tau-\sum_{\tau=t}^{\mathcal{T}}\left(F^\top\right)^{\tau-t}Pg_\theta^{(\tau-t+1)}(c_t).
\end{equation*}
\end{theorem}

\begin{proof}
    
By Theorem 3.2 in~\cite{yu2021powerpredictionsonlinecontrol}, given predictions $\hat{w}_{t:\mathcal{T}|t}$, the MPC 
 solution $u_t^{\textsc{MPC}}$ to the problem defined in~\eqref{eq:mpc_formulation} is
    \begin{equation*}
        u_t^{\textsc{MPC}} = -(R+B^{\top}PB)^{-1} B^{\top}\Big(PAx_t + \sum_{\tau = t}^{\mathcal{T}} (F^{\top})^{\tau - t} P \hat{w}_{\tau|t}\Big ),
    \end{equation*}
where $F \coloneqq A - B(R + B^{\top}PB)^{-1}B^{\top}PA \coloneqq A - BK$.
Thus, applying Lemma \ref{lem:regret} twice, we obtain
    \begin{equation*}
        J(\theta_{1:T})-J(\theta^{\star})=\sum_{t=0}^{T-1}\psi_t(\theta_t)^\top H\psi_t(\theta_t)-\psi_t(\theta^{\star})^\top H\psi_t(\theta^{\star}),
    \end{equation*}
where $\psi_t(\theta)$ is defined in Theorem~\ref{thm:main}. For notational simplicity, we define $\psi_t\coloneqq \psi_t(\theta_t)$ and $\phi_t\coloneqq \psi_t(\theta^{\star})$.
% where we define
%     \begin{equation*}
%         \psi_t\coloneqq\sum_{\tau=t}^{T-1}\left(F^\top\right)^{\tau-t}Pw_\tau-\sum_{\tau=t}^{\mathcal{T}}\left(F^\top\right)^{\tau-t}Pg_{\theta_t}^{(\tau-t+1)}(c_t)
%     \end{equation*}
%  and 
%  \begin{equation*}
%      \phi_t\coloneqq\sum_{\tau=t}^{T-1}\left(F^\top\right)^{\tau-t}Pw_\tau-\sum_{\tau=t}^{\mathcal{T}}\left(F^\top\right)^{\tau-t}Pg_{\theta^{\star}}^{(\tau-t+1)}(c_t).
%  \end{equation*}
By our model assumption, $Q,R\succ 0$ implies $H\succeq 0$, using the convexity of $x^\top Hx$,
\begin{align} 
    &\frac{1}{2}\sum_{t=0}^{T-1}\psi_t^\top H\psi_t-\phi_t^\top H\phi_t \leq\sum_{t=0}^{T-1}\psi_t^\top H(\psi_t-\phi_t)  =\nonumber\\
    \nonumber
    &\sum_{t=0}^{T-1}\psi_t^\top H \left(\sum_{\tau=t}^{\mathcal{T}}\left(F^\top\right)^{\tau-t}P\left(g_{\theta^{\star}}^{(\tau-t+1)}(c_t)-g_{\theta_t}^{(\tau-t+1)}(c_t)\right)\right).
\end{align} 

Applying Assumption~\ref{asp:convex}, since $g_{\theta}$ is affine in $\theta$, continuing from above,
\begin{align}
\nonumber
&J(\theta_{1:T})-J(\theta^{\star})\leq\\
&-2\sum_{t=0}^{T-1}\psi_t^\top H \left(\sum_{\tau=t}^{\mathcal{T}}\left(F^\top\right)^{\tau-t}P\nabla_\theta g_{\theta_t}^{(\tau-t+1)}(c_t)^\top\right)(\theta_t-\theta^{\star})\notag \\
&=\sum_{t=0}^{T-1}\nabla_{\theta_t}(\psi_t^\top H \psi_t)^\top(\theta_t-\theta^{\star}),
\label{eq:linear regret}
\end{align} 
where we denote $\nabla_{\theta_t}(\psi_t^\top H \psi_t)\coloneqq\nabla_{\theta_t}(\psi_t^\top H \psi_t)|_{\theta=\theta_t}$, and have used Assumption \ref{asp:convex}. 
Denote the gradient of the loss function as $\nabla_\theta L_{t-k+1}(\theta_{t-k+1})$ as $l_{t-k+1}$. It follows that the RHS of~\eqref{eq:linear regret} can be rewritten as 
\begin{align}
\nonumber
&\sum_{t=0}^{T-1}\nabla_{\theta_t}(\psi_t^\top H \psi_t)^\top(\theta_t-\theta^{\star}) \notag\\
        =&\sum_{t=0}^{T-1}l_{t-k+1}^\top(\theta_t-\theta^{\star})+\sum_{t=0}^{T-1}(l_t-l_{t-k+1})^\top(\theta_t-\theta^{\star}) \notag\\
        &\quad + \sum_{t=0}^{T-1}(\nabla_{\theta_t}(\psi_t^\top H \psi_t)^\top-l_t^\top)(\theta_t-\theta^{\star}).\label{eq:decomposed_regret}
\end{align}
The bound in~\eqref{eq:loss_discrepancy_bound} is obtained by bounding the terms in~\Cref{eq:decomposed_regret} separately, and the details are relegated to Appendix~\ref{app:proof_main}.
% and the fact that $g_{\theta^{\star}}^{(\tau-t+1)}(c_t)-g_{\theta_t}^{(\tau-t+1)}(c_t)=\nabla_\theta g_{\theta_t}^{(\tau-t+1)}(c_t)^\top(\theta_t-\theta^{\star})$ to derive \eqref{eq:linear regret}.
\end{proof}

Furthermore, for a wide range of loss functions $(L_t:t\in [T])$, as examplified in Corollary \ref{cor} below, they exhibit a bounded discrepancy from the MPC cost with the discrepancy decaying exponentially as $k$ increases, i.e., $\emph{\texttt{LD}}(L_t,\psi_t^\top H\psi_t)\leq C\rho^k$. 

\begin{corollary}\label{cor} Under Assumption \ref{asp:mpc}, for the case when $L_t(\theta)=\hat{\psi}_t(\theta)^\top H\hat{\psi}_t(\theta)$ with
\[
    \hat{\psi}_t(\theta) \coloneqq \sum_{\tau=t}^{\mathcal{T}}\left(F^\top\right)^{\tau-t}Pw_\tau-\sum_{\tau=t}^{\mathcal{T}}\left(F^\top\right)^{\tau-t}Pg_\theta^{(\tau-t+1)}(c_t),
\]
there exist constants $C$ and $0<\rho<1$ such that 
    \[
    \emph{\texttt{LD}}(L_t,\psi_t^\top H\psi_t)\leq C\rho^k, \ \text{ for all } t\in [T],
    \]
    where
   $\psi_t(\theta)$ is defined in Theorem~\ref{thm:main}.
\end{corollary}

For the particular choice of $L_t$ in Corollary~\ref{cor}, the term 
$
D\sum_{t=0}^{T-1}{\texttt{LD}}(L_t,\psi_t^\top H \psi_t)
$
in the bound~\eqref{eq:loss_discrepancy_bound} of Theorem~\ref{thm:main} simplifies to $CDT\rho^k$. This provides an insightful guideline for selecting $k$. Specifically, by setting 
$
k = \log_{1/\rho} T,
$
we can achieve a regret bound of $O(\sqrt{T\log T})$.

%\begin{theorem}
%    \begin{align*}
%\text{ALG}-\text{OPT}\leq &2\lambda^2\|H\|(2\beta^2\epsilon_w+2(1-\beta)^2W)\\
%&+2(1-\lambda)^2\|R+B^TPB\|\epsilon_u
%\end{align*}
%where $H:=B(R+B^TPB)^{-1}B^T$, ALG denotes the cost achieved by algorithm \ref{alg:lqr}, OPT %denotes the offline optimal cost %and $F:=A−B(R+B^TPB)^{-1}B^TPA$
%$$\epsilon_w:=\sum_{t=0}^{T-1}\|\sum_{\tau=t}^{T-1}(F^T)^{\tau-t}P(\hat w_t-w_t)\|^2$$
%$$\epsilon_{u}:=\sum_{t=0}^{T-1}\|u_{t,\mathcal{M}}-u_{t|x_t}^*\|^2$$
%$$W:=\sum_{t=0}^{T-1}\|\sum_{\tau=t}^{T-1}(F^T)^{\tau-t}Pw_t\|^2$$

%\begin{align*}
%\epsilon_w&=\sum_{t=0}^{T-1}\|\beta\sum_{\tau=t}^{T-1}(F^T)^{\tau-t}P(\hat w_t-w_t)-(1-\beta)\sum_{\tau=t}^{T-1}(F^T)^{\tau-t}Pw_t\|^2\\
%            &\leq 2\sum_{t=0}^{T-1}\beta^2\|\sum_{\tau=t}^{T-1}(F^T)^{\tau-t}P(\hat w_t-w_t)\|^2+(1-\beta)^2\|\sum_{\tau=t}^{T-1}(F^T)^{\tau-t}Pw_t\|^2
%\end{align*}
%\end{theorem}
\section{Numerical Experiment} \label{sec:experiment}
In this section, we present a detailed demonstration of how the L2D module is fine-tuned in Subsection~\ref{sec:finetune}. We implement the L2D module with LLaMa-8B~\cite{touvron2023llamaopenefficientfoundation}, a causal language model. Causal language models predict the next token in a sequence and return prediction scores for each vocabulary token, which aligns well with L2D’s requirement of generating scenario probabilities. To obtain these probabilities, we use the following prompt template: ``\textit{\texttt{[task description]}. Given the current context: \texttt{[context description]}, the expected scenario is \texttt{[MASK]}.}'' For a predefined set of scenarios~($\mathcal{S}$ in Subsection~\ref{sec:l2d}), we evaluate the prediction score assigned by the model by inserting each candidate token into the \textit{\texttt{[MASK]}} position. These scores are then normalized using a softmax function to produce a probability distribution.

In Subsection~\ref{sec:context_impact} and Subsection~\ref{sec:real_world_app}, we demonstrate the effectiveness of \IMPC framework in two scenarios. The first is a robot tracking task, where an agent is required to follow an unknown trajectory. In this setting, the target position is revealed only immediately before the control decision is made. Moreover, environmental disturbances, such as wind force, are described textually and passed through the L2D module to assist the MPC controller. We plot the resulting trajectory to directly demonstrate the impact of contextual information. The second scenario is a real-world energy management task, in which human-provided instructions, such as news and scheduled appointments, are used to predict external disturbances. In this case, we further extend \IMPC~to operate on a nonlinear dynamical system, we also adopt a piecewise linear cost instead of the quadratic cost defined in~\cref{eq:quadratic_costs}, demonstrating the generalizability of \IMPC in diverse system dynamics and different cost functions.

\subsection{Fine-Tuning L2D Module} \label{sec:finetune}

\textbf{Fine-Tuning L2D Module.} We first define the preferences among scenarios based on observed disturbances. Given the the prediction horizon $k$, and $\mathcal{T} \coloneqq \min(t + k - 1, T - 1)$, after each $k$ control steps, the true disturbance sequence $w_{t:\mathcal{T}}$ is revealed. For each scenario $s \in \mathcal{S}$, we measure the distance between $w_{t:\mathcal{T}}$ and $w_{t : \mathcal{T}}^s$ as:
\begin{equation*}
    \mathrm{dist}(w_{t:\mathcal{T}}, w_{t : \mathcal{T}}^s) \coloneqq \| w_{t:\mathcal{T}} - w_{t : \mathcal{T}}^s\|.
\end{equation*}
Given the context $c_t$ at time $t$, we say scenario $s_1$ is more likely to occur than scenario $s_2$, denoted as $s_1 \succ s_2$, if $\mathrm{dist}(w_{t:\mathcal{T}}, w_{t : \mathcal{T}}^{s_1}) < \mathrm{dist}(w_{t:\mathcal{T}}, w_{t : \mathcal{T}}^{s_2})$. This yields a dataset of pairwise comparisons:
\begin{equation*}
    \mathcal{D} = \{(c_t, s_{w,t}, s_{l,t}) | t \in \mathcal{T}_{\text{obs}}; s_w, s_l \in \mathcal{S}; s_i \neq s_j\},
\end{equation*}
where $\mathcal{T}_{\text{obs}}\in \{1, \ldots, T\}$ denotes the subset of time step at which observations and comparisons are collected, $s_{w,t}$ and $s_{l,t}$ denote the preferred and dispreferred scenario at time $t$ respectively. In practice, the parameter adjustment loop operates at a different timescale than the control loop, we tune the L2D module in batches. During both the robot tracking and energy management experiments, the L2D module is fine-tuned using DPO~\cite{rafailov2023direct} once the size of pairwise comparison dataset \(\mathcal{D}\) reaches a predefined threshold. Specifically, the parameter \(\theta\) is updated when \(|\mathcal{D}| = 240\) in the robot tracking task and \(|\mathcal{D}| = 500\) in the energy management scenario.

\subsection{Application 1: Robot Tracking} \label{sec:context_impact}
\textbf{Problem description.} We examine a 2D robotic tracking scenario~\cite{Li_2022, yu2021powerpredictionsonlinecontrol} where a controller follows an unknown, fixed astroid path, which is 
\begin{equation*}
    y_t\coloneqq
    \begin{bmatrix}
2\sin^3 (t/38.2) \\
2\cos^3(t/38.2)
\end{bmatrix}, \quad t\in [T].
\end{equation*}
Let $p_t\in \mathbb{R}^2$ and $v_t\in \mathbb{R}^2$ denote the location and the velocity of robot at time $t$. $p_{t+1}$ depends on previous location $p_{t}$ and velocity $v_t$ such that $p_{t+1}=p_t+0.2v_t$. The velocity is influenced by both the control input $u_t$ and the external disturbances, specifically the wind force represented by $Z_t$. Let $x_t\coloneqq q_t-y_t$ represent the deviation between the robot's trajectory and the target trajectory. The system dynamics can be represented by:
\begin{equation*}
    \begin{bmatrix}
x_{t+1} \\
v_{t+1}
\end{bmatrix}
=
A\begin{bmatrix}
x_{t} \\
v_{t}
\end{bmatrix}
+Bu_t+w_t,
\end{equation*}
where
\begin{equation*}
    A\coloneqq
    \begin{bmatrix}
        1 & 0 &0.2 &0 \\
        0 &1 &0 &0.2\\
        0 &0 &1 &0\\
        0 &0 &0 &1
    \end{bmatrix},
    B\coloneqq
    \begin{bmatrix}
        0 & 0  \\
        0 &0 \\
        0.2 &0 \\
        0 &0.2
    \end{bmatrix},
\end{equation*}
\begin{align*}
    &w_t\coloneqq Ay_t-y_{t+1}+ Z_t, Z_t\coloneqq\begin{bmatrix}
        0 \! \! \!  &0 \! \! \! 
 &-0.2Z_{t}^{(1)} \! \! \!  &-0.2Z_{t}^{(2)}
    \end{bmatrix}^\top\! \! \!,
\end{align*}
where $Z_t$ is a random variable representing the random wind force. Specifically, during the experiment, we randomly select a subset of $[T]$, denote as $T_{\text{sub}} = \{20, 25, 32, 35, 40, 84, 133, 145, 158,215\}$, and let
\begin{equation*}
    \begin{cases}
        Z^{(1)}_t, Z^{(2)}_t \sim \mathcal{U}(-45, 45) & t \in T_{\text{sub}}\\
        Z^{(1)}_t, Z^{(2)}_t \sim \mathcal{U}(-2, 2) & \text{otherwise}
    \end{cases}
\end{equation*}

To track the trajectory, the controller uses quadratic cost in~\Cref{eq:quadratic_costs} with
\begin{equation*}
    Q\coloneqq
    \begin{bmatrix}
        1 & 0 &0 &0 \\
        0 &1 &0 &0\\
        0 &0 &0 &0\\
        0 &0 &0 &0
    \end{bmatrix}, \ \ 
    R\coloneqq
    \begin{bmatrix}
        10^{-2} & 0\\
        0 &10^{-2}
    \end{bmatrix}.
\end{equation*}

\textbf{Experimental results.}
In our first experiment, we plot the robot trajectory to demonstrate the impact of external contextual information, in Fig.~\ref{fig:robot_tracking}, a zoom-in figure is used to demonstrate the exact impact of instructions,
We compare three settings:
\begin{enumerate}
    \item \textit{Classic MPC}: An MPC baseline using \eqref{eq:mpc_formulation} where disturbance predictions are set to zero (\(\hat{w}_{t:\mathcal{T}|t} = 0\)) for all \(t\), equivalent to a Linear Quadratic Regulator (LQR) for the problem~\eqref{eq:quadratic_costs}-\eqref{eq:sys_constraints}.
    \item \textit{Zero-Shot Prompting}: MPC with predictions generated by a zero-shot prompted LLaMA-8B~\cite{touvron2023llamaopenefficientfoundation} model.
    \item \textit{\IMPC (with fine-tuning)}: Our proposed method using LLaMA-8B~\cite{touvron2023llamaopenefficientfoundation} fine-tuned in a closed-loop with DPO~\cite{rafailov2023direct}.
\end{enumerate}

We observe that the baseline~(classic MPC without L2D predictions) tends to deviate from the trajectory due to its inability to predict disturbances from external information. In contrast, both methods with L2D predictions improve the performance of the robot, with the fine-tuned version demonstrating a closer robot path to the designed trajectory.

\begin{figure*}[ht]
    \centering
    \vspace{2mm}\includegraphics[width=0.8\linewidth]{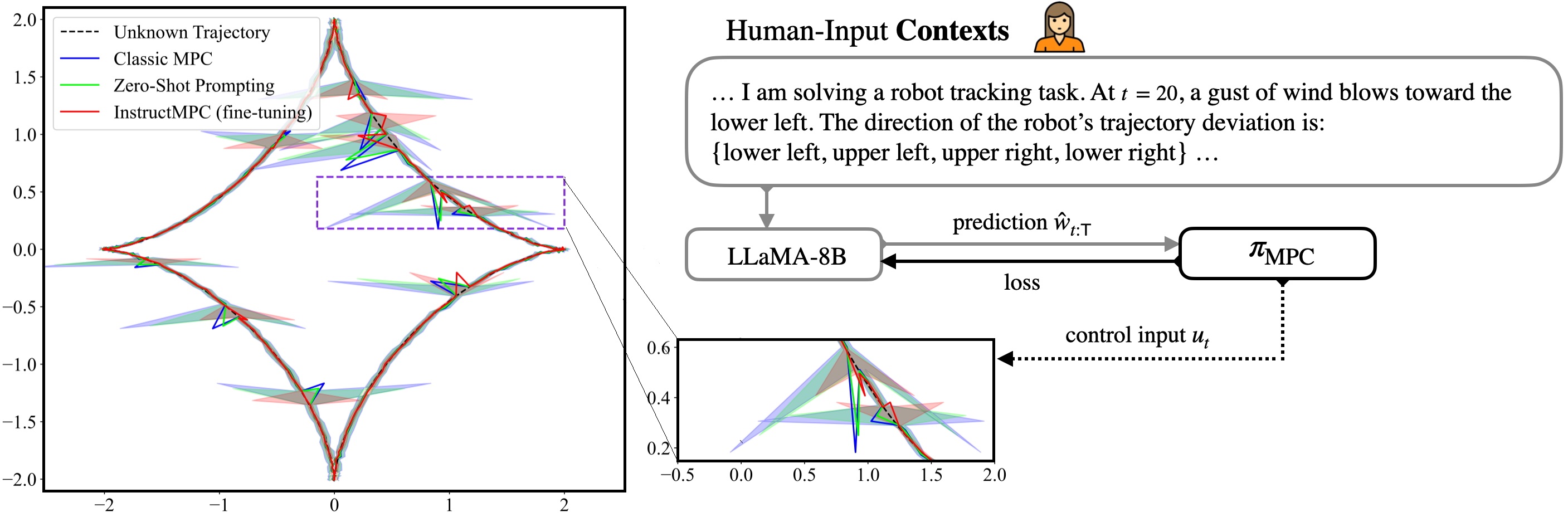}
    \caption{\textbf{Application 1:} \textit{Robot Tracking (\Cref{sec:context_impact})}. Average robot tracking trajectories of three methods: 1. MPC in~\eqref{eq:mpc_formulation} with predictions $\hat{w}_{t:\mathcal{T}|t}=0$ for all $t$ (classic MPC); 2. MPC with predictions from a zero-shot prompting LLaMA-8B~\cite{touvron2023llamaopenefficientfoundation} model; 3. \IMPC with LLaMA-8B~\cite{touvron2023llamaopenefficientfoundation} fine-tuned in a closed-loop. For each method, we ran $20$ episodes, and plot the average robot trajectory. The shadow region denotes the variance of trajectories.}
    \label{fig:robot_tracking}
\end{figure*}

\subsection{Application 2: Building Energy Management} \label{sec:real_world_app}
\textbf{Problem description.}
We consider the problem of energy management for a campus building. The system consists of a battery for energy storage and a series of photovoltaic~(PV) arrays to harvest the sunlight. Let $x_t \in [0, 1]$ denote the state of charge~(SoC) of the battery at time $t$, and $u_t \in [0, 1]$ represent the charging rate, $u_t = 1$ corresponds to charging at the maximum allowable power, and $u_t = 0$ indicates no charging. The nonlinear system dynamics can be represented by:
$
    x_{t+1} = f(x_t, u_t) + w_t,
$
where \( f: \mathbb{R}^2 \to \mathbb{R} \) captures the system’s inherent battery dynamics, $w_t \in \mathbb{R}$ represents the external disturbance, which can be attributed to the variations in solar irradiance or fluctuations in electricity demand. These disturbances are inferred from contextual information, including weather forecasts, human activity schedules, and natural language instructions. Examples of such contextual information include: ``\textit{The CDC conference deadline is approaching, and many people are running experiments,}'' or ``\textit{I will be fine-tuning my LLaMA-1B model at 3 PM today, which will take approximately two hours.}'' The objective of the energy management system is to minimize the total cost of purchasing electricity from the grid. This cost is computed using a piecewise linear function that reflects real-world peak and valley electricity prices.

\begin{figure*}[ht]
    \centering
    \vspace{2mm}
\includegraphics[width= 0.9\linewidth]{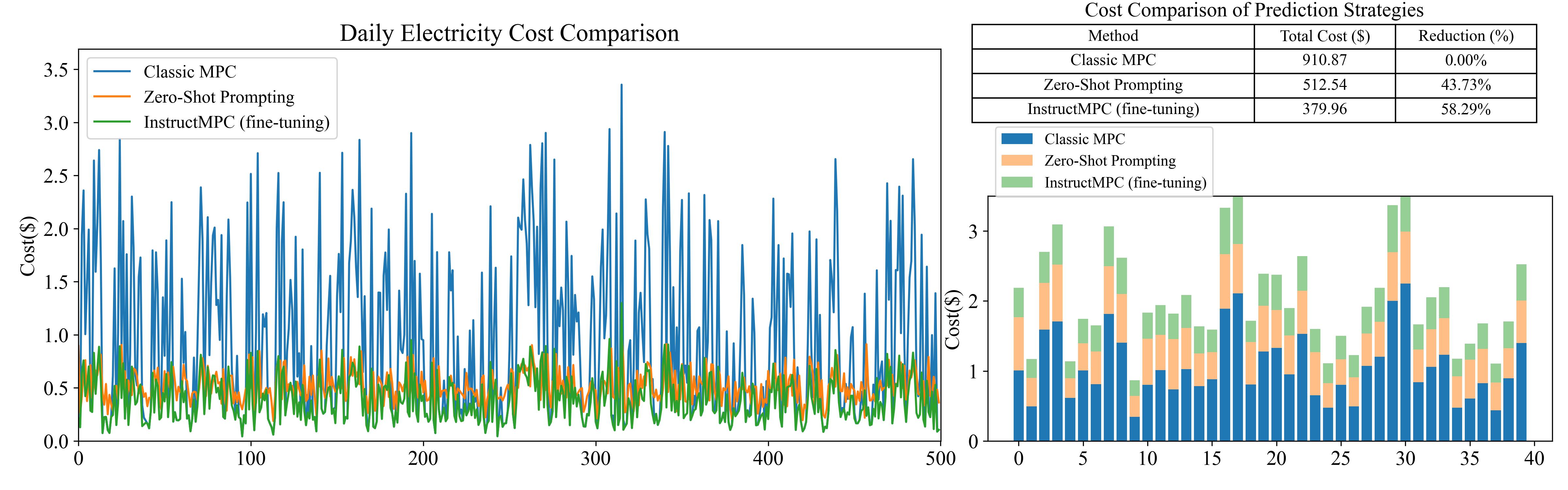}
\caption{\textbf{Application 2:} \textit{Building Energy Management (\Cref{sec:real_world_app})}. Daily electricity cost comparison of three approaches: 1. MPC in~\eqref{eq:mpc_formulation} with predictions $\hat{w}_{t:\mathcal{T}|t}=0$ for all $t$ (classic MPC); 2. MPC with predictions from a zero-shot prompting LLaMA-8B~\cite{touvron2023llamaopenefficientfoundation} model; 3. \IMPC with LLaMA-8B~\cite{touvron2023llamaopenefficientfoundation} fine-tuned in a closed-loop. }
    \label{fig:energy_management}
\end{figure*}

\textbf{Experimental results.}
We show the performance of \IMPC\ in an energy management scenario using a synthetic workload and real-world weather data. Each energy-consuming task arrives at a constant rate of $1$ task per unit time, and its energy demand is drawn from a uniform distribution \(\mathcal{U}(1, 5)\), with each unit corresponding to $0.5$~kWh. For PV generation, we randomly sample $500$ days of real weather data and model the corresponding PV generation using PVLIB simulator~\cite{anderson2023pvlib}, and simulate three different MPC settings described above.

The results are presented in Fig.~\ref{fig:energy_management}. The left part shows the daily electricity costs over $500$ days for all three methods, demonstrating the effectiveness of incorporating contextual predictions and model parameter adjustment. The upper-right part provides a summary table comparing the total electricity cost and the corresponding percentage reduction, while the bottom-right part provides a closer view of daily costs over $40$ days, further illustrating the cost savings by \IMPC.

\section{Conclusion}

We have presented the \IMPC framework that integrates real-time human instructions into Model Predictive Control (MPC) via a human-LLM-in-the-loop approach. This work demonstrates how leveraging contextual information can significantly improve MPC’s predictive accuracy and performance in complex, real-world applications. By dynamically generating context-aware disturbance predictions and refining them using advanced techniques like Proximal Policy Optimization (PPO) and Direct Preference Optimization (DPO), \IMPC offers improved adaptability and performance across a wide range of applications. For linear dynamics, we established a robust performance guarantee, proving that IMPC achieves a regret bound of $O(\sqrt{T\log T})$. In future work, we aim to extend these results to nonlinear systems, and plan to conduct more experiments under different key parameters, such as the prediction horizon $k$ and the choice of LLM. We will also deploy the \IMPC framework on our newly deployed physical energy testbed OpenCEM~\cite{10.1145/3679240.3734678} to assess real-world generalization beyond synthetic workloads.

% In future work, we aim to extend these results to nonlinear systems, further explore the integration of richer contextual cues to enhance adaptability and robustness in complex environments, and analyze the effect of the L2D module's parameters.~\cite{10.1145/3679240.3734678}

\bibliography{biblography}
\bibliographystyle{ieeetr}

% \vspace{-5pt}

\makeatletter
\renewcommand\appendix{
  \par
  \setcounter{section}{0}
  \setcounter{subsection}{0}
  \gdef\thesection{\@Alph\c@section}
}
\makeatother

\appendix
 
\vspace{-10pt}

\subsection*{\bf Appendix A: Proof of Theorem~\ref{thm:main}}
\label{app:proof_main}

% the explicit solution of \eqref{eq:mpc_formulation} and 

The following lemma characterizes the quadratic cost gap.
% \begin{lemma}[Theorem 3.2 in~\cite{yu2021powerpredictionsonlinecontrol}]
%     Given predictions $\hat{w}_{t:\mathcal{T}|t}$, the solution $u_t$ to the problem defined in~\eqref{eq:mpc_formulation} is
%     \begin{equation*}
%         u_t^{\textsc{MPC}} = -(R+B^{\top}PB)^{-1} B^{\top}\Big(PAx_t + \sum_{\tau = t}^{\mathcal{T}} (F^{\top})^{\tau - t} P \hat{w}_{\tau|t}\Big ),
%     \end{equation*}
% where $F \coloneqq A - B(R + B^{\top}PB)^{-1}B^{\top}PA \coloneqq A - BK$.
% \end{lemma}

\begin{lemma}[Lemma 13 in~\cite{yu2022competitivecontroldelayedimperfect}]\label{lem:regret}
For any $\psi_t\in \mathbb{R}^n$, if at each time $t\in [T]$, a controller $\pi$ implements a control input
$
        u_t^{\pi}=-(R+B^\top PB)^{-1}B^\top\big(PAx_t+\sum_{\tau=t}^{T-1}\left(F^\top\right)^{\tau-t}Pw_\tau-\psi_t\big),
$
    then the difference between the optimal cost $J^{\star}$ and the algorithm cost $J(\pi)$ is given by $
       J(\pi)-J^{\star} = \sum_{t=0}^{T-1}\psi_t^\top H\psi_t,$
where $H \coloneqq B(R+B^\top PB)^{-1}B^\top$ and $F \coloneqq A-HPA$.
\end{lemma}

Now, we bound the terms in~\Cref{eq:decomposed_regret} separately:
\begin{align}
\nonumber
&\sum_{t=0}^{T-1}\nabla_{\theta_t}(\psi_t^\top H \psi_t)^\top(\theta_t-\theta^{\star}) \notag\\
        =&\underbracket{\sum_{t=0}^{T-1}l_{t-k+1}^\top(\theta_t-\theta^{\star})}_{\text{(a)}}+\underbracket{\sum_{t=0}^{T-1}(l_t-l_{t-k+1})^\top(\theta_t-\theta^{\star})}_{\text{(b)}} \notag\\
        &\quad + \sum_{t=0}^{T-1}(\nabla_{\theta_t}(\psi_t^\top H \psi_t)^\top-l_t^\top)(\theta_t-\theta^{\star}).
        \nonumber
\end{align}
First, based on the gradient update rule in~\eqref{eq:theta_update} of \IMPC,
 \begin{align}
&\|\theta_{t+1}-\theta^{\star}\|^2 \notag \\
=& \|\theta_t-\eta_t l_{t-k+1}-\theta^{\star}\|^2 \notag\\
= & \|\theta_t-\theta^{\star}\|^2-2\eta_t l_{t-k+1}^\top(\theta_t-\theta^{\star})+\eta_t^2\|l_{t-k+1}\|^2.
\label{eq:bound_theta_1}
\end{align}
Then, rearranging the terms in~\Cref{eq:bound_theta_1} and noting that the gradients are bounded by Assumption~\ref{asp:lg},
    \begin{align*}
        2l_{t-k+1}^\top(\theta_t-\theta^{\star})\leq \frac{\|\theta_t-\theta^{\star}\|^2-\|\theta_{t+1}-\theta^{\star}\|^2}{\eta_t}+\eta_tG^2.
    \end{align*}
Summing above from $t=0$ to $T-1$, for (a),
    \begin{align}
        2&\sum_{t=0}^{T-1}l_{t-k+1}^\top(\theta_t-\theta^{\star}) \notag\\
        \leq &\sum_{t=0}^{T-1}\frac{\|\theta_t-\theta^{\star}\|^2-\|\theta_{t+1}-\theta^{\star}\|^2}{\eta_t}+\eta_tG^2 \notag\\
        \leq &\sum_{t=0}^{T-1}\|\theta_t-\theta^{\star}\|^2\left(\frac{1}{\eta_t}-\frac{1}{\eta_{t-1}}\right)+G^2\sum_{t=0}^{T-1}\eta_t \notag\\
        \leq &\frac{D^2}{\eta_{T-1}}+G^2\sum_{t=0}^{T-1}\eta_t .\label{eq: bound1}
    \end{align}
    %\begin{align}
    %    2&\sum_{t=0}^{T-1}(l_t-l_{t-k+1})^\top(\theta_t-\theta^{\star})\leq
    %    2\sum_{t=0}^{T-1}\|l_t-l_{t-k+1}\|\|\theta_t-\theta^{\star}\|\\
    %    &\leq 2D\sum_{t=0}^{T-1}\|l_t-l_{t-k+1}\|\\
    %    &\leq 2D\sum_{t=0}^{T-1}\|\nabla_\theta L(\theta_{t},c_{t})-\nabla_\theta L(\theta_{t},c_{t-k+1})\|\\
     %   &+\|\nabla_\theta L(\theta_{t},c_{t-k+1})-\nabla_\theta L(\theta_{t-k+1},c_{t-k+1})\|\\
     %   &\leq 2DL\sum_{t=0}^{T-1}\|c_t-c_{t-k+1}\|+\|\theta_t-\theta_{t-k+1}\|
    %\end{align}

Furthermore, the second term in~\Cref{eq:decomposed_regret} can be bounded from above as
    \begin{align}
      \text{(b)}=  &\sum_{t=0}^{T-1}(l_t-l_{t-k+1})^\top(\theta_t-\theta^{\star}) \notag\\
      =&\sum_{t=0}^{T-k}l_t^\top(\theta_t-\theta_{t+k-1})+\sum_{t=T-k+1}^{T-1}l_t^\top(\theta_t-\theta^{\star}) \notag\\
        %=&\sum_{t=0}^{T-1}l_t^\top(\theta_t-\theta_{t+k-1})+(l_{T-1}-l_{-k+1})^\top(-\theta^{\star}) \notag\\
        \leq &\sum_{t=0}^{T-k}(G\|\theta_t-\theta_{t+k-1}\|) +(k-1)GD \notag\\
        =&(k-1)GD+G\sum_{t=0}^{T-k}\left\|\sum_{\tau=t}^{t+k-2}\eta_{\tau}l_{\tau-k+1}\right\| \notag\\
        \leq& (k-1)GD+G^2\sum_{t=0}^{T-k}\sum_{\tau=t}^{t+k-2}\eta_{\tau} \notag\\
        \leq &(k-1)GD+(k-1)G^2\sum_{t=0}^{T-1}\eta_t. \label{eq:bound2}
    \end{align}
Recalling \eqref{eq:linear regret}, 
    \begin{align}
        &J(\theta_{1:T})-J(\theta^{\star}) \notag\\
        \leq &\frac{D^2}{2\eta_{T-1}}+\frac{G^2}{2}\sum_{t=0}^{T-1}\eta_t +(k-1)GD+(k-1)G^2\sum_{t=0}^{T-1}\eta_t \notag\\
        +&\sum_{t=0}^{T-1}\|\nabla_{\theta_t}(\psi_t^\top H \psi_t)-l_t\|\|\theta_t-\theta^{\star}\| \label{eq:-1}\\ 
        \leq &\frac{D^2}{2\eta_{T-1}}+(k-\frac{1}{2})G^2\sum_{t=0}^{T-1}\eta_t+(k-1)GD \notag\\
        +&D\sum_{t=0}^{T-1}\emph{\texttt{LD}}(L_t,\psi_t^\top H\psi_t),\label{eq:fin}
    \end{align}
where we have used~\eqref{eq: bound1} and \eqref{eq:bound2} to derive~ \eqref{eq:-1}; the last inequality~\eqref{eq:fin} holds by the loss discrepancy in~\Cref{def:df}. Finally, applying $\eta_t={D}/\left({G\sqrt{2(2k-1)(t+1)}}\right)$ to \eqref{eq:fin},
    \begin{align}
        &J(\theta_{1:T})-J(\theta^{\star}) \notag\\
        = &GD\sqrt{(k-\frac{1}{2})T} +\sqrt{k-\frac{1}{2}}\frac{GD}{2}\sum_{t=0}^{T-1}\frac{1}{\sqrt{t+1}}\notag\\
        +&D\sum_{t=0}^{T-1}\emph{\texttt{LD}}(L_t,\psi_t^\top H\psi_t)+(k-1)GD \notag\\
        \leq &2GD\sqrt{(k-\frac{1}{2})T}+D\sum_{t=0}^{T-1}\emph{\texttt{LD}}(L_t,\psi_t^\top H\psi_t)+(k-1)GD. \nonumber
    \end{align}
The inequality above in the second statement of Theorem~\ref{thm:main} is obtained by using $\sum_{t=1}^T\frac{1}{\sqrt{t}}\leq 2\sqrt{T}$.

%~\eqref{eq:fin2}

\subsection*{\bf Appendix B: Proof of Corollary~\ref{cor}}

Suppose $L_t(\theta)=\hat{\psi}_t(\theta)^\top H\hat{\psi}_t(\theta)$ where
\[
    \hat{\psi}_t(\theta) \coloneqq \sum_{\tau=t}^{\mathcal{T}}\left(F^\top\right)^{\tau-t}Pw_\tau-\sum_{\tau=t}^{\mathcal{T}}\left(F^\top\right)^{\tau-t}Pg_\theta^{(\tau-t+1)}(c_t).
\]

By Definition~\ref{def:df},
  \begin{align}
        \emph{\texttt{LD}}(L_t,\psi_t^\top H\psi_t)&=\left\|\frac{\partial (\psi_t^\top H \psi_t)}{\partial \theta}-\frac{\partial (\hat\psi_t^\top H \hat\psi_t)}{\partial \theta}\right\| \notag \\
        &=\left\|2(\psi_t-\hat \psi_t)^\top H \frac{\partial  \psi_t}{\partial \theta}\right\| 
        \label{eq:LFB_0}
        \\
        &\leq 2\|H\|\left\|\psi_t-\hat \psi_t\right\|\left\|\frac{\partial  \psi_t}{\partial \theta}\right\|.\label{eq:LFD}
    \end{align}
We have use the fact that $\frac{\partial  \psi_t}{\partial \theta}=\frac{\partial  \hat\psi_t}{\partial \theta}$ to get~\eqref{eq:LFB_0}. Note that if $\mathcal{T}=T-1$, then $\psi_t=\hat \psi_t$, $\emph{\texttt{LD}}(L_t,\psi_t^\top H \psi_t)=0$. Thus, it suffices to restrict our analysis to the case when $\mathcal{T}=t+k-1$. It follows that
\begin{align}
    \left\|\psi_t-\hat \psi_t\right\|&=\left\|\sum_{\tau=t+k}^{T-1}\left(F^\top\right)^{\tau-t}Pw_\tau\right\| \notag \\
        &\leq W\|P\|\sum_{\tau=t+k}^{T-1}\|F\|^{\tau-t} \label{eq:psi-hpsi_0} \\
        &\leq W\frac{\|P\|\|F\|^k}{1-\|F\|}, \label{eq:psi-hpsi}
    \end{align}
where~\eqref{eq:psi-hpsi_0} follows from the Assumption \ref{asp:mpc} so that $\|w_t\|\leq W$ for all $t\in [T]$. Moreover, by Assumption \ref{asp:convex}, $\|\nabla_\theta g_{\theta_t}^{(\tau-t+1)}(c_t)\|\leq L$,
\begin{align}
    \left\|\frac{\partial  \psi_t}{\partial \theta}\right\|&=\left\|\sum_{\tau=t}^{t+k-1}\left(F^\top\right)^{\tau-t}P\nabla_\theta g_{\theta_t}^{(\tau-t+1)}(c_t)\right\| \notag\\
        &\leq L\|P\|\sum_{\tau=t}^{t+k-1}\|F\|^{\tau-t} \notag \leq L\frac{\|P\|}{1-\|F\|}.\label{eq:nablapsi}
    \end{align}

Finally, combining \eqref{eq:LFD},\eqref{eq:psi-hpsi}, and~\eqref{eq:nablapsi} we have
    \begin{align}
        \emph{\texttt{LD}}(L_t,\psi_t^\top H\psi_t)\leq 2\frac{LW\|P\|^2\|H\|\|F\|^k}{(1-\|F\|)^2}.
    \end{align}

% \subsection{Example Prompts and Scenarios}
% \begin{table}[ht]
%     \centering
%     \begin{tabular}{c|c}
%          &  \\
%          & 
%     \end{tabular}
%     \caption{Caption}
%     \label{tab:placeholder}
% \end{table}
\end{document}